\documentclass[sigconf]{aamas}  

\usepackage{booktabs}

\usepackage{graphicx} 

\usepackage{algorithm,algorithmic}
\usepackage{booktabs, subcaption} 
\usepackage{xcolor, color} 

\usepackage{amsmath,amsthm,amssymb}
\begin{document}
%
\title{A Learning and Masking Approach to Secure Learning}
\author{Linh Nguyen, Sky Wang, Arunesh Sinha\\
University of Michigan, Ann Arbor\\
\texttt{\{lvnguyen,skywang,arunesh\}@umich.edu}\\
}
\keywords{adversarial examples; robust learning} 

\begin{abstract}
Deep Neural Networks (DNNs) have been shown to be vulnerable against adversarial examples, which are data points cleverly constructed to fool the classifier. Such attacks can be devastating in practice, especially as DNNs are being applied to ever increasing critical tasks like image recognition in autonomous driving. In this paper, we introduce a new perspective on the problem. We do so by first defining robustness of a classifier to adversarial exploitation. Next, we show that the problem of adversarial example generation can be posed as learning problem. We also categorize attacks in literature into high and low perturbation attacks; well-known attacks like FGSM~\cite{fgsm} and  our attack produce higher perturbation adversarial examples while the more potent but computationally inefficient Carlini-Wagner~\cite{carlini2017towards} (CW) attack is low perturbation. Next, we show that the dual approach of the attack learning problem can be used as a defensive technique that is effective against high perturbation attacks. Finally, we show that a classifier masking method achieved by adding noise to the a neural network's logit output protects against low distortion attacks such as the CW attack. We also show that both our learning and masking defense can work simultaneously to protect against multiple attacks. We demonstrate the efficacy of our techniques by experimenting with the MNIST and CIFAR-10 datasets. 
\end{abstract}

\maketitle

\section{Introduction}
Recent advances in deep learning have led to its wide adoption in various challenging tasks such
as image classification. However, the current state of the art has
been shown to be vulnerable to \emph{adversarial examples}, small perturbations of the original inputs, often indistinguishable to a human, but 
carefully crafted to misguide the learning models into producing incorrect outputs. Recent results have shown that generating
these adversarial examples are inexpensive~\cite{fgsm}.
Moreover, as safety critical applications such as autonomous driving increasingly rely on these tasks, it is
imperative that the learning models be reliable and secure against such adversarial examples. 


Prior work has yielded a lot of attack methods that generate adversarial samples, and defense techniques that improve 
the accuracy on these samples (see related work for details). 
However, defenses are often specific to certain attacks and cannot adaptively defend against any future attack and some general defense techniques  have been shown to be ineffective against more powerful novel attacks. More generally, attacks and defenses have followed the cat-and-mouse game that is typical of many security settings. Further, traditional machine learning theory assumes a fixed stochastic environment hence accuracy in the traditional sense is not enough to measure performance in presence of an adversarial agent.

In this paper, with the goal of generality, we pursue a principled approach to attacks and defense. Starting from a theoretical robustness definition, we present a attack and a defense that learns to generate adversarial examples against any given classifier and learns to defend against any attack respectively. Based on formal intuition, we categorize known attacks into high and low perturbation attacks. Our \emph{learning} attack is a high perturbation attack and analogously our \emph{learning} defense technique defends against high perturbation attack. For low perturbation attacks, we provide a \emph{masking} approach that defends against such attacks. Our two defense techniques can be combined to defend against multiple types of attacks. While our guiding principle is general, this paper focuses on the specific domain of adversarial examples in image classification.


Our \emph{first contribution} is a definition of \emph{robustness of classifiers} in presence of an adversarial agent. Towards the definition, we define the exploitable space by the adversary which includes data points already mis-classified (errors) by any given classifier and any data points that can be perturbed by the adversary to force mis-classifications. Robustness is defined as the probability of data points occurring in the exploitable space. We believe our definition captures the essence of the multi-agent defender-adversary interaction, and is natural as our robustness is a strictly stronger concept than accuracy. We also analyze why accuracy fails to measure robustness. The formal set-up also provides an intuition for all the techniques proposed in this paper.

Our \emph{second contribution} is an \emph{attack learning neural network} (ALN). ALN is motivated by the fact that adversarial examples for a given classifier $C$ are subsets of the input space that the $C$ mis-classifies. Thus, given a data distribution with data points $x$ and a classifier $C$ trained on such data, we train a feed forward neural network $A$ with the goal of generating output points $A(x)$ in the mis-classified space. Towards this end, we re-purpose an autoencoder to work as our ALN $A$ with a special choice of loss function that aims to make (1) the classifier $C$ mis-classify $A(x)$ and (2) minimize the difference between $x$ and $A(x)$. 

Our \emph{third contribution} are two defense techniques: \emph{defense learning neural network} (DLN) and \emph{noise augmented classifier} (NAC). Following the motivation and design of ALN, we motivate DLN $D$ as a neural network that, given any classifier $C$ attacked by an attack technique $A$, takes in an adversarial example $A(x)$ and aims to generate benign example $D(A(x))$ that \emph{does not} lie in the mis-classified space of $C$. The DLN is prepended to the classifier $C$ acting as a sanitizer for $C$.
 Again, similar to the ALN, we re-purpose an autoencoder with a special loss function suited for the goal of the DLN. For non-adversarial inputs the DLN is encouraged to reproduce the input as well as make the classifier predict correctly. 
 We show that DLN allows for attack and defense to be set up as a repeated competition leading to more robust classifiers. Next, while DLN works efficiently for attacks that produces adversarial examples with high perturbation, such as fast gradient sign method~\cite{fgsm} (FGSM), it is not practical for low perturbation attacks (discussed in details in Section~\ref{repeatDLN}) such as Carlini-Wagner~\cite{carlini2017towards} (CW). For low perturbation attacks, we present NAC which masks the classifier boundary by adding a very small noise at the logits output of the neural network classifier. The small noise added affects classification in rare cases, thereby ensuring original accuracy is maintained, but also fools low perturbation attacks as the attack is mislead by the \emph{incorrect} logits. DLN and NAC can work together to defend simultaneously against both high and low perturbation attacks.


We tested our approach on two datasets: MNIST and CIFAR-10. Our ALN based attack was able to attack all classifiers we considered and achieve performance comparable to other high perturbation attacks. Our defense approach made the resultant classifier robust to the FGSM and CW. Detailed experiments are presented in Section~\ref{section4} and~\ref{section5}. Missing proofs are in an online appendix\footnote{\url{https://drive.google.com/open?id=1CBaHsU6IL9jQ4UN_2yYUteUsbLP2aKMk}} (see footnote).

\section{Attack Model}
Given the adversarial setting, it is imperative to define the capabilities of the adversary, which we do in this section. First, we use \emph{inference phase} of a classifier to mean the stage when the classifier is actually deployed as an application (after all training and testing is done). The attacker attacks \emph{only} in the inference phase and can channel his attack \emph{only} through the inputs. In particular, the attacker cannot change the classifier weights or inject any noise in the hidden layers. The attacker has access to the classifier weights, so that it can compute gradients if required. The attacker's goal is to produce adversarial data points that get mis-classified by the classifier. These adversarial examples should be legitimate (that is not a garbage noisy image) and the true class and the predicted class of the data point could be additional constraints for the adversary.

\section{Approach}
This section formally describes our approach to the adversarial example generation and defense problem using the notion of robustness we define. We start by defining basic notations. Let the 
function $C: X \rightarrow Y$ denote a classifier that takes input data points with feature values in $X$ and outputs a label among the 
possible $k$ labels $Y = \{1, \ldots, k\}$. Further, for neural networks based classifiers we can define $C_p: X \rightarrow \Delta Y$ as the function that 
takes in data and produces a probability distribution over labels. Thus, $C = \max\{C_p (x)\}$, where $\max$ provides the maximum component of the 
vector $C_p (x)$. Let $\overline{sim}(x,x')$ denote the dissimilarity between $x$ and $x'$. Let $H(p,q)$ denote the cross entropy $- \sum_i p_i \log (q_i)$. In 
particular, let $H(p)$ denotes the entropy given by $H(p,p)$. For this paper, we assume $X$ is the set of legitimate images (and not garbage images or ambiguous images). Legitimate images are different for different domains, e.g., they are digits for digit classification. Given a label $y$, let $Cat(y)$ denote the categorical probability distribution with the component for $y$ set to $1$ and all else $0$. Let $\overline{opsim}(\mathbf{y},\mathbf{y}')$ denote the dissimilarity between output distributions $\mathbf{y},\mathbf{y}' \in \Delta Y$.

\subsection{Robustness}
We first introduce some concepts from PAC learning~\cite{Anthony2009NNL}, in order to present the formal results in this section. It is assumed that data points arise from a fixed but unknown distribution $\mathcal{P}$ over $X$. We denote the probability mass over a set $Z \subset X$ as $\mathcal{P}(Z)$. A loss function $l(y_x, C(x))$ captures the loss of predicting $C(x)$ when the true label for $x$ is $y_x$. As we are focused on classification, we restrict ourselves to the ideal $0/1$ loss, that is, $1$ for incorrect classification and $0$ otherwise. A classifier $C$ is chosen that minimizes the empirical loss over the $n$ training data points $\sum_{i=1}^n l(y_{x_i},x_i)$. Given enough data, PAC learning theory guarantees that $C$ also minimizes the expected loss $\int_X l(y_x, C(x)) \mathcal{P}(x)$. Given, $0/1$ loss this quantity is just $\mathcal{P}(M_C(X))$, where $M_C(X) \subset X$ denote the region where the classifier $C$ mis-classifies.  Accuracy for a classifier is then just $1 - \mathcal{P}(M_C(X))$. In this paper we will assume that the amount of data is always enough to obtain low expected loss. Observe that a classifier can achieve high accuracy (low expected loss) even though its predictions in the low probability regions may be wrong.

All classifier families have a capacity that limits the complexity of separators (hypothesis space) that they can model. A higher capacity classifier family can model more non-smooth separators\footnote{While capacity is defined for any function class~\cite{Anthony2009NNL} (includes deep neural networks), the value is known only for simple classifiers like single layered neural networks.}. Previous work~\cite{fgsm} has conjectured that adversarial examples abound due to the low capacity of the classifier family used. See Figure~\ref{fig:intuition}A for an illustration.

\textbf{Adversarial exploitable space}: Define $E_{C,\epsilon}(X) = M_C(X) \cup \{x~|~ \overline{sim}(x,M_C(X)) \leq \epsilon\}$, where $\overline{sim}$ is a dissimilarity measure that depends on the domain and $\overline{sim}(x,M_C(X))$ denotes the lowest dissimilarity of $x$ with any data point in $M_C(X)$. For image classification $\overline{sim}$ can just be the $l_2$ (Euclidean) distance: $\sqrt{\sum_i (x_i - x'_i)^2}$ where $i$ indexes the pixels. $E_{C,\epsilon}(X)$ is the adversarial exploitable space, as this space includes all points that are either mis-classified or can be mis-classified by a minor $\epsilon$-perturbation. Note that we assume that any already present mis-classifications of the classifier is exploitable by the adversary without the need of any perturbation. For example, if a stop sign image in a dataset is mis-classified then an adversary can simply use this image as is to fool an autonomously driven vehicle.

\textbf{Robustness}: Robustness is simply defined as $1 - \mathcal{P}(E_{C,\epsilon}(X))$. First, it is easy to see that robustness is a strictly stronger concept than accuracy, that is, a classifier with high robustness has higher accuracy. We believe this property makes our definition more natural than other current definitions. Further, another readily inferable property from the definition of $E_{C,\epsilon}$ that we utilize later is that a classifier $C'$ with $M_{C'}(X) \subset M_C(X)$ is  more robust than classifier $C$ in the same setting. We call a classifier $C'$ perfect if the robustness is 100\%. 

There are a number of subtle aspects of the definition that we elaborate upon below:
\begin{itemize}
\item A 100\% robust classifier can still have  $M_{C'}(X) \neq \phi$. This is because robustness is still defined w.r.t. $\mathcal{P}$, for example, large compact regions of zero probability with small sub-region of erroneous prediction far away from the boundary can still make robustness 100\%. However, $M_{C'}(X) = \phi$ provides 100\% robustness for any $\mathcal{P}$. Thus, robustness based on just $M_{C'}(X)$ and not $\mathcal{P}$ is a stronger but much more restrictive concept of robustness than ours.
\item A perfect classifier (100\% robust) is practically impossible due to large data requirement especially as the capacity of the classifier family grows. As shown in Figure~\ref{fig:intuition} low capacity classifiers cannot model complex separators, thus, large capacity is required to achieve robustness. On the other hand, classifiers families with large capacity but not enough data tend to overfit the data~\cite{Anthony2009NNL}. Thus, there is a delicate balance between the capacity of the classifier family used and amount of data available. The relation between amount of data and capacity is not very well understood for Dep Neural Networks. In any case, perfect robustness provides a goal that robust classifiers should aim to achieve.
In this paper, for the purpose of defense, we seek to increase the robustness of classifiers.
\item Robustness in practice may apparently seem to be computable by calculating the accuracy for the test set and the adversarially perturbed test set for any given dataset, which we also do and has been done in all prior work. However, this relies on the fact that the attack is all powerful, i.e., it can attack \emph{all} perturb-able points. It is easy to construct abstract examples with probability measure zero mis-classification set (single mis-classified point in a continuous Euclidean space) that is computationally intractable for practical attacks to discover. A detailed analysis of computing robustness is beyond the scope of this paper and is left for future work.
\item The definition can be easily extended to weigh some kinds of mis-classification more, if required. For example, predicting a malware as benign is more harmful than the opposite erroneous prediction. For our focus area of image classification in this paper, researchers have generally considered all mis-classification equally important. Also the $\overline{sim}$ function in the definition is reasonably well agreed upon in literature on adversarial learning as the $l_2$ distance; however, we show later in experiments that $l_2$ distance does not capture similarity well enough. Instantiating the definition for other domains such as malware classification requires exploring $\overline{sim}$ further such as how to capture that two malwares are functionally similar.
\end{itemize}

Lastly, compared to past work~\cite{wang2016theoretical,fawzi2015fundamental}, our robustness definition has a clear relation to accuracy and not orthogonal to it. Also, our definition uses the ideal 0/1 loss function rather than an approximate loss function $l$ (often used in training due to smoothness) as used in other definitions~\cite{madry2017towards,cisse2017parseval}. We posit that the 0/1 loss measures robustness more precisely, as these other approaches have specified the adversary goal as aiming to perturb in order to produce the maximum loss within an $\epsilon$ ball $B(x,\epsilon)$ of any given point $x$, with the defender expected loss defined as $\int_X \max_{z \in B(x,\epsilon)} l(y_x, C(z)) \mathcal{P}(x)$. However, this means that even if the class is same throughout the $\epsilon$ ball, with a varying $l$ the adversary still conducts a supposed ``attack'' and increases loss for the defender without flipping labels. For example, the well-known hinge loss varies rapidly within one of the classes and such supposed attacks could lead to an overestimation of the loss for defender and hence underestimate robustness. Further, use of an approximation in the definition allows an adversary to bypass the definition by exploiting the approximation by $l$ when the true loss is 0/1. It is an interesting question for future on what kind of approximations can help in computing robustness within reasonable error bounds.

Finally, we analyze if accuracy is ever suitable to capture robustness. First, we make a few mild technical assumptions that there exists a density $p(x)$ for the data distribution $\mathcal{P}$ over $X$, $X$ is a metric space with metric $d$ and $vol(X) = 1$. We have the following result:
\begin{theorem} \label{thm1}
$1 - a$ accuracy implies at least $1 - (a + \nu + K\epsilon/T)$ robustness for any output $C$ if
\begin{itemize}
    \item For all $x \in X$, $\overline{sim}(x,x') \geq T d(x,x')$ for some $T > 0$.
    \item $M_C(X)$ lies in a low density region, that is, for all $x \in M_C(X)$ we have $p(x) \leq \nu$ for some small $\nu$.
    \item $p(x)$ is $K$-Lipschitz, that is, $|p(x) - p(x')| \leq K d(x,x')$ for all $x,x' \in X$.
\end{itemize}
\end{theorem}

The first two conditions in the above result are quite natural. In simple words, the first two conditions says dissimilarity increases with distance (high $T$) and the regions that the output classifier predicts badly has low amount of data in the data-set (low $\nu$). 

However, the final condition may not be satisfied in many natural settings. This condition states that the data distribution must not change abruptly (low $K$). This is required as the natural behavior of most classifiers is to predict bad in a low data density region and if this region is near a high data density region, the adversary can successfully modify the data points in the high density region causing loss of robustness. But in high dimensional spaces, data distribution is quite likely to be not distributed smoothly with many pockets or sub-spaces of zero density as pointed out in a recent experimental work~\cite{tramer2017space}. 
Thus, data distribution is an important contributing factor that determines robustness.

\subsection{ALN}
Our goal is to train a neural network ALN to produce samples in the misclassification region of a given neural network based classifier. The ALN acts on a data point $x$ producing $x'$. Thus, we choose the following loss function for the ALN that 
takes into account the output for the given classifier: 
$$\alpha \overline{sim}(x, x') - \overline{opsim}(Cat(y_x), C_p(x')) \; ,$$
The input dissimilarity term in the loss function aims to produce data points $x'$ that are similar to the original input $x$ while the output dissimilarity term aims to 
maximize the difference between the true label of $x$ and prediction of $C$ on $x'$. The $\alpha$ is a weight that is tuned through a simple search. Observe that this loss function is general and can be used with any classifier (by inferring $C_p$ from $C$ in case of specific non neural network based classifiers). 
For the image classification problem we use the $l_2$ distance for $\overline{sim}$. For $\overline{opsim}$ we tried a number of functions, but the best performance was for the $l_1$ loss $||Cat(y_x) - C_p(x')||_1$.

Note that an alternate loss function is possible that does not use the actual label $y_x$ of $x$, rather  using $C_p(x)$. This would also work assuming that the classifier is good; for poor classifiers a lot of the data points are as it is mis-classified and hence adversarial example generation is interesting only for good classifiers. Further, this choice would allow using unlabeled data for conducting such an attack, making attack easier for an attacker. However, in our experiments we use the more powerful attack using the labels.


Next, we provide a formal intuition of what ALN actually achieves. 
Any adversarial example generation can be seen as a distribution transformer $F$ such that acting on the data distribution $\mathcal{P}$ the resultant distribution $F(\mathcal{P})$ has support mostly limited to $M_C(X)$. The support may not completely limited to $M_C(X)$ as the attacks are never 100\% effective. Also, attacks in literature aim to find points in $M_C(X)$ that are close to given images in the original dataset. ALN is essentially a neural network representation of such a function $F$ against a given classifier $C$. See Figure~\ref{fig:intuition}B for an illustration. We return to this interpretation in the next sub-section to provide a formal intuition about the DLN defense.

Lastly, we show in our experiments that ALN produces adversarial examples whose perturbations are roughly of the same order as the prior attack FGSM. We categorize these as high perturbation attacks. On the other the attack CW produces adversarial perturbation with very small perturbations, we call such attacks low perturbation attacks. As mentioned earlier, we provide two separate defenses for these two types of attacks. Both these defense can be used simultaneously to defend against both types of attacks.

\begin{figure}[t]
     \centering \includegraphics[scale=0.25]{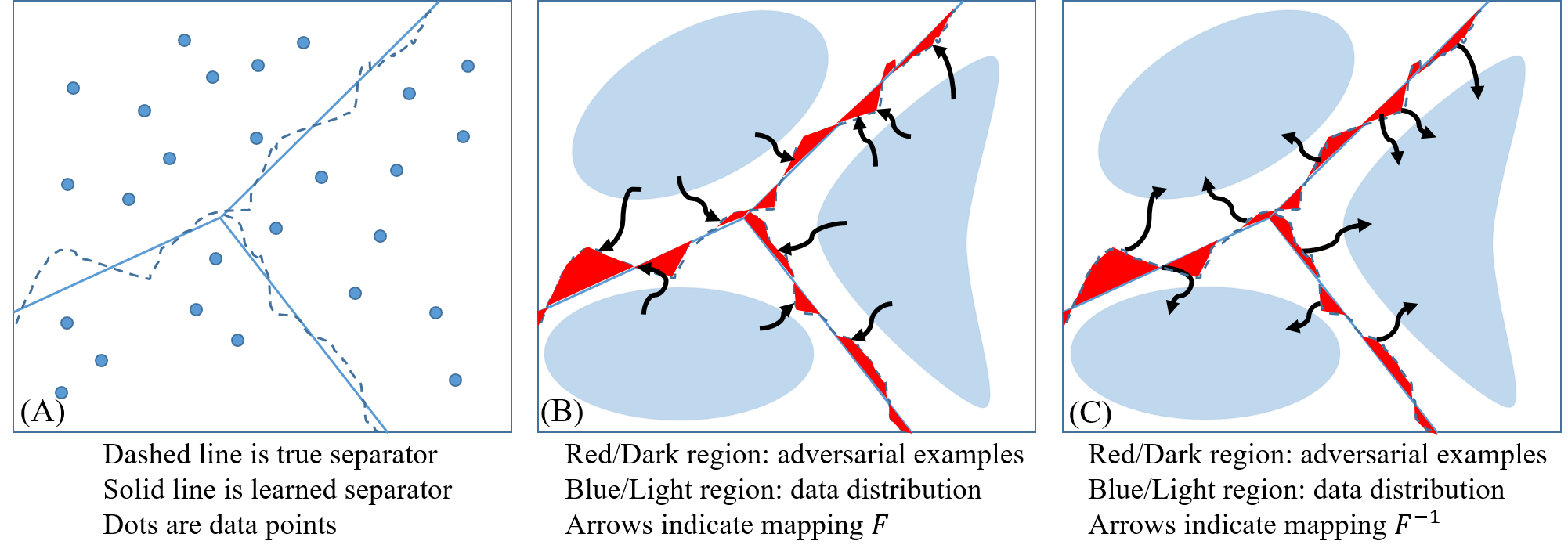}
     \caption{Intuition behind ALN and DLN. (A) shows a linear classifier (low capacity) is not able to accurately model a non-linear boundary. (B) shows the ALN as the distribution  mapping function $F$. (C) shows that DLN does the reverse mapping of ALN.}
     \label{fig:intuition}
 \end{figure}

\subsection{DLN}
Our first defense approach is to insert a neural network DLN $D$ between the input and classifier so that $D$ sanitizes the input enabling the classifier to correctly classify the input. Each data point for training DLN has three parts: $x'$ is the image to sanitize, $x$ is the expected output and $y_x$ is the correct label of $x$. $x$ are images from the provided dataset, and there are two possibilities for $x'$: (1) $x' = x$ so that DLN attempts to satisfy $C(D(x)) = y_x$, even if $C(x) \neq y_x$, and (2) $x' = A(x)$ so that DLN undoes the attack and make the classifier $C$ correctly classify $x'$.

We formulate a loss function for DLN that, similar to ALN, has two terms: $\overline{sim}(x, D(x'))$ that aims to produce output $D(x')$ close to  $x$ and $\overline{opsim}(Cat(y_x),C_p(D(x')))$ that aims to make the classifier output on $D(x')$ be the same as $y_x$. Thus, the loss function is
$$
\alpha \overline{sim}(x, D(x')) + \overline{opsim}(Cat(y_x),C_p(D(x'))) \; .
$$
In this paper we only use $\alpha = 1$.
Note that the attack $A$ is used as a black box here to generate training data and is not a part of the loss function. After training the DLN, our new classifier is $C'$ which is $C$ prepended by the DLN. The working of DLN can be interpreted as an inverse map $F^{-1}$ for the mapping $F$ induced by the attack $A$. See Figure~\ref{fig:intuition}C for an illustration. For the image classification problem we use the $l_2$ distance for $\overline{sim}$. For $\overline{opsim}$ we tried a number of functions, but the best performance was for the cross-entropy loss $H(Cat(y_x), C_p(D(x'))$. 

An important point to note is that the original classifier $C$ is unchanged. What this ensures is the mis-classification space $M_C(X)$ does not change and allows us to prove an important result about $C'$ under certain assumptions. For the sake of this result, we assume that attacks $A$ generate adversarial examples in a sub region $M_{C,A}(X) \subset M_C(X)$. We also assume a good DLN $D$, that is, $C(D(x))$ is correct for a non-empty subset $Z \subset M_{C,A}(X)$ and $C(D(x))$ continues to be correct for all $x \notin M_{C}(X)$. Then, we prove 
\begin{lemma} \label{lemma1}
Assuming $M_{C,A}(X) \subset M_C(X)$, DLN is good as defined above, and $M_{C,A}(X) \neq \phi$, then
$M_{C'}(X) \subset M_C(X)$
\end{lemma}
\begin{proof}
Since DLN does not decrease the performance of $C$ on points outside $M_C(X)$, $C'$'s prediction on inputs outside $M_{C}(X)$ is correct, hence $M_{C'}(X) \subseteq M_{C}(X)$.
Any data point not mis-classified by a classifier does not belong to its mis-classification space. Good sanitization by DLN makes $C'$ predict correctly on $Z \subset M_{C,A}(X)$, which makes $M_{C,A}(X) \cap M_{C'}(X) \subset M_{C,A}(X)$. Thus, we can claim the result in the lemma statement.
\end{proof}
While the above proof is under ideal assumptions, it provides an intuition to how the defense works. Namely, the reduction in the adversarial exploitable space makes the new classifier more robust (see robustness properties earlier). This also motivates the generalization of this technique to multiple attacks presented in the next sub-section.


\subsection{Repeated DLN Against Multiple Attacks} \label{repeatDLN}
The above DLN can be naturally extended to multiple attacks, say $A_1, \ldots, A_n$. The only change required is to feed in all possible adversarial examples $A_1(x)$'s$, \ldots,$$A_n(x)$'s. It is straightforward to see that under assumptions of Lemma~\ref{lemma1} for all the attacks, the resultant classifier $C'$ has an adversarial example space $M_{C'}(X)$ that removes subsets of $M_{C,A_i}(X)$ for all $A_i \in \mathcal{A}$ from $M_{C}(X)$. This provides, at least theoretically under ideal assumptions, a monotonic robustness improvement property with increasing number of attacks for the DLN based approach. 
In fact, if all the attacks combined act as a generator for all the points in $M_C(X)$, then given enough data and perfect sanitization the resultant classifier $C'$ tends towards achieving $M_{C'}(X) = \phi$ which essentially would make $C'$ a perfect classifier.
Perfect classifiers have no adversarial examples. 

However, attacks rarely explore all of the mis-classified space, which is why new attacks have defeated prior defense techniques. Even for our approach, attacks successfully work against the DLN that has been trained only once (accuracy numbers are in Experiments). However, DLN allows for easy retraining (without retraining the classifier) as follows: repeatedly attack and re-learn a DLN in rounds, that is, conduct an attack on the classifier obtained in every round and train a DLN in a round using the attacked training data from all the previous rounds and the original training data. More formally, at round $i$ our training data consists of $i$ copies of original training data and $i$ instances of attacked training data from previous rounds. Observe that we add copies of the original training data in each round, this is because the adversarial data swamps out the original training data and accuracy suffers in regions where the original training data is distributed. See Figure~\ref{fig:dlnvsnac} for an illustration of how repeated DLN works.

\begin{figure}[t]
     \centering \includegraphics[scale=0.35]{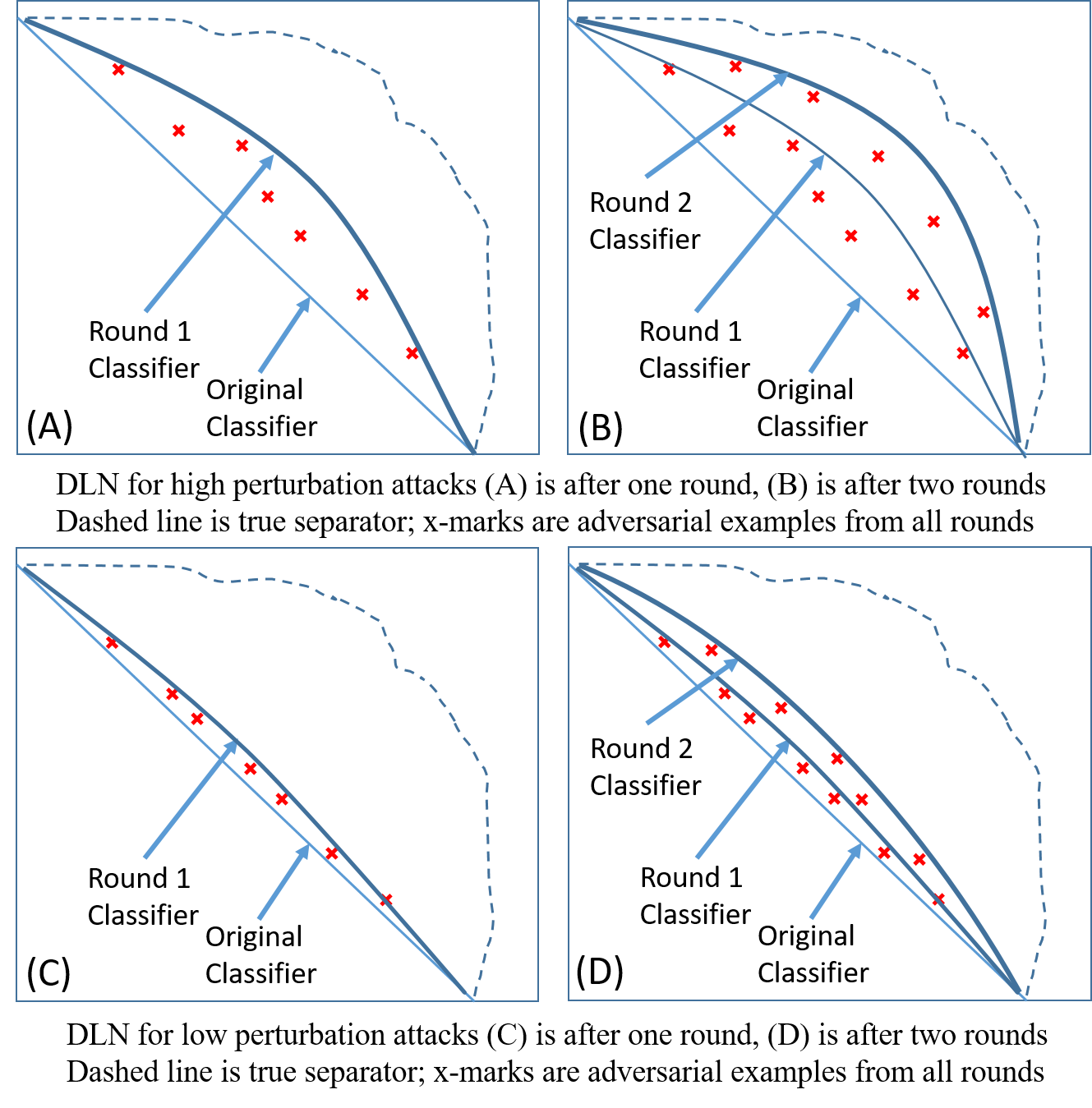}
     \caption{Intuition behind working of repeated DLN against high and low perturbation attacks. (A),(B) shows a high perturbation attack causes a faster improvement in resultant classifier. Further, beyond some rounds the attack does not work as it can only find adversarial examples with high perturbation. (C),(D) shows a low perturbation attack causes a slow improvement in resultant classifier.}
     \label{fig:dlnvsnac}
 \end{figure}

The following result provides formal intuition for this approach:
\begin{lemma}
Assume the following conditions hold for every round $i$: $M_{C_{i-1},A}(X) \subset  M_{C_{i-1}}(X)$ and the DLN $D_i$ has good memory, which means that given there exists a largest set $Z_i \subset M_{C_{i-1},A}(X)$ which the DLN $D_i$ correctly sanitizes so that $C(D_i(x))$ is correct for all $x \in Z_i $ then $Z_{i-1} \subset Z_i$. That is DLN $D_i$ can correctly sanitize data points that the previous round DLN did plus possibly more data points.
Further, $C(D_i(x))$ continues to be correct for all $x \notin M_{C}(X)$.
Then the classifier $C_n$ after $n$ rounds satisfies $M_{C_n}(X) \subset M_{C_{n-1}}(X)$.
\end{lemma}
\begin{proof}
Arguing similarly to Lemma~\ref{lemma1} we can show that $M_{C_n}(X) \subseteq M_{C}(X)$ due to the correct classification outside of $M_{C}(X)$. 
Further, it is easily inferable that $M_{C_n}(X) = M_{C}(X) \backslash Z_n$ given $Z_n$ is a subset of $M_C(X)$ and given the largest such set condition on $Z_i$. Then, the good memory property leads to the required result.
\end{proof}
The attack-defense competition technique is somewhat akin to GANs~\cite{goodfellow2014generative}. However, there is a big difference, since in every round the dataset used to train the DLN grows. Practically, this requires DLN to have a large capacity in order to be effective; also depending on the capacity and the size of dataset over or under fitting problems could arise, which needs to be taken care of in practice. Also, the training become more expensive over rounds with increasing data size. In particular, low perturbation attacks are not defeated with few rounds. We do observe improvement with the low perturbation CW attack over rounds, but the improvement is very small, as represented visually in Figure~\ref{fig:dlnvsnac}. The main reason is that low perturbation attacks only exposes a very small volume of misclassified space, thus, it would require a huge number of rounds for repeated DLN to reduce the mis-classified space to such a small volume that cannot be attacked. This motivates our next approach of noise augmented classifier.

\subsection{NAC}
Figure~\ref{fig:dlnvsnac} also provides a hint on how to overcome low perturbation attacks. In order to achieve low perturbation, such attacks rely a lot on the exact classifier boundary. Thus, masking the classifier boundary can fool low perturbation attacks. We achieve this by adding a small noise to the logits of the neural network calling the resultant classifier a noise augmented classifier (NAC). This noise should be small enough that it does not affect the classification of original data points by much, but is able to  mis-lead the low perturbation attack. Also, following our intuition NAC should not provide any defense against high perturbation attacks, which we indeed observe in our experiments. However, observe that DLN and NAC can be used simultaneously, thus, providing benefits of both defense which we show in our experiments. 

Further, a natural idea to bypass the defense provided by NAC is to take the average of multiple logit outputs for the same given input image (to cancel the randomness) and then use the average logits as the logits required for the CW attack. We show experimentally that this idea does not work effectively even after averaging over many logit outputs.

\section{Experiments for Attacks}
\label{section4}
All our experiments, including for DLN and NAC, were conducted using the Keras framework on a NVIDIA K40 GPU.  The learning problem can be solved using any gradient-based optimizer. In our case, we used Adam with learning rate 0.0002. We use two well-known datasets: MNIST digits and CIFAR-10 colored images.

We consider two classifiers one for MNIST and one for CIFAR-10: we call them $C_M$ and $C_C$. These classifiers are variants of well-known architectures that achieved state-of-the-art performances on their respective datasets. As stated earlier, we consider three attacks: ALN, FGSM and CW. CW has been referred to in the literature~\cite{xu2017feature} as one of the best attacks till date (at the time of writing of this paper), while FGSM runs extremely fast. For the autoencoder we use a fourteen hidden layer convolutional architecture. Our code is publicly available, but the github link is elided in this submitted version for the review process.
For our experiments we pre-process all the images so that the pixels values lie between $[-0.5, 0.5]$, so all components (attacks, autoencoders, classifiers) work in space $[-0.5, 0.5]$. We use FGSM with values of $0.03$ and $0.01$ for its parameter $\epsilon$ on MNIST and CIFAR, respectively.

Observe that all these attacks work against a given classifier $C$, thus, we use the notation $A(C,.)$ to denote the attack $A$ acting on an image $x$ to produce the adversarial example $A(C,x)$ ($A$ can be any of the three attacks). $A(C,Z)$ denotes the set of adversarial examples $\{A(C,x)~|~ x \in Z\}$.
We report accuracies on various test sets: (1) original test dataset ($OTD$): these images are the original test dataset from the dataset under consideration, (2) $A(C,OTD)$ is the adversarially perturbed dataset using attack $A$ against classifier $C$, for example, this could be FGSM$({C_M},OTD)$. We also report distortion numbers as has been done in literature~\cite{carlini2017towards}. Distortion measures how much perturbation on average was added by an attack and is meant to capture how visually similar the image is to the original image. Distortion is measured as the average over all test images of the $l_2$ distance between the original and perturbed image.



\textbf{Results}:
Untargeted attacks refers to attack that aim to produce mis-classification but not with the goal of making the classifier output a specific label. Targeted attacks aim to produce an adversarial example that gets classified as a given class $y$. It is also possible to modify ALN to perform targeted attacks. This is achieved by modifying the loss function to use a positive $\overline{opsim}$ term, like the DLN loss function, but using the target class label $y$ instead of the original class label $y_x$ in the $\overline{opsim}$ term. Then, we can perform an ALN attack in two ways: ALNU uses the ALN loss function as stated and ALNT constructs a targeted attack per class label differing from the original label and chooses the one with least distortion. Figure~\ref{fig:targeted} shows an example of targeted attack with different target labels. Table~\ref{untargetmnist} shows this approach for MNIST with the targeted ALNT version performance better than other attacks.

\begin{figure}[t]
    \centering \includegraphics[width=230px,height=32px]{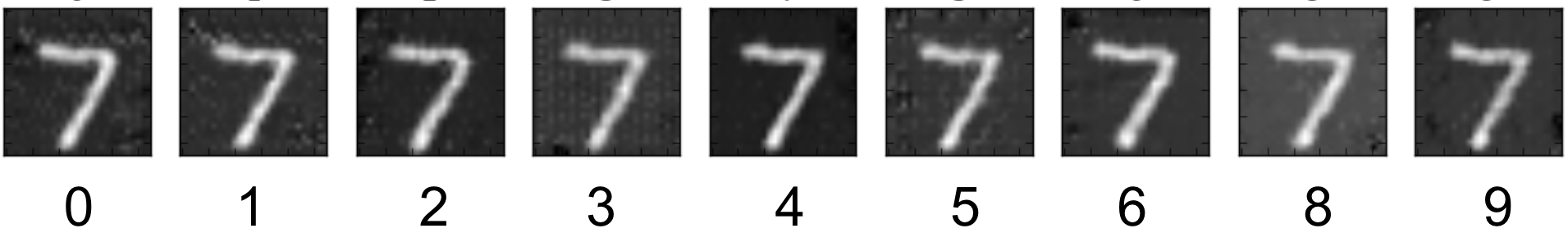}
    \caption{Targeted attacks by ALN: target class on bottom} \label{fig:targeted}
\end{figure}

\begin{table}[h!]
\centering
\begin{tabular}{c c c}
\toprule
Test data type & Accuracy & Distortion\\
\midrule
$OTD$ & 99.45 \% & $-$\\
FGSM$({C_M},OTD)$ & 0.72 \% & 14.99 \\
CW$({C_M},OTD)$ & 0.03 \% & 1.51 \\
ALNU$({C_M},OTD)$ & 1.65 \% & 4.43 \\
ALNT$({C_M},OTD)$ & 0.0 \% & 4.34 \\
\bottomrule
\end{tabular}
\caption{Attacks on MNIST Dataset} \label{untargetmnist}
\end{table}

Table~\ref{untargetcifar} shows the result of untargeted attacks using ALN, FGSM and CW on the CIFAR-10 dataset. We can see that ALN, just like FGSM, produces slightly higher adversarial accuracy for MNIST, but the distortion of FGSM is much higher. This does result in a large difference in visual quality of the adversarial examples produced---see Figure~\ref{fig:untargetedcifar} for randomly chosen 25 perturbed images using ALN and Figure~\ref{fig:fgsm_cifar} for randomly chosen 25 perturbed images using FGSM. Also, we can attribute the higher distortion of ALN for CIFAR as compared to MNIST partially to the CIFAR being a higher dimensional space problems and the same capacity of the autoencoder we used for CIFAR and MNIST.

\begin{table}[h!]
\centering
\begin{tabular}{c c c}
\toprule
Test data type & Accuracy & Distortion\\
\midrule
$OTD$ & 84.59 \% & $-$\\
FGSM$({C_C},OTD)$ & 4.21 \% & 10.03 \\
CW$({C_C},OTD)$ & 0 \% & 0.18 \\
ALNU$({C_C},OTD)$ & 6.16 \% & 2.57 \\
ALNT$({C_C},OTD)$ & 8.21 \% & 3.01 \\
\bottomrule
\end{tabular}
\caption{Attacks on CIFAR-10 Dataset} \label{untargetcifar}
\end{table}

\begin{figure}[t]
    \centering \includegraphics[width=150px,height=150px]{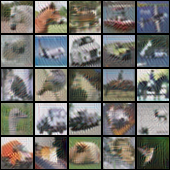}
    \caption{Untargeted attack by ALN for CIFAR-10} \label{fig:untargetedcifar}
\end{figure}

\begin{figure}[t]
    \centering \includegraphics[width=150px,height=150px]{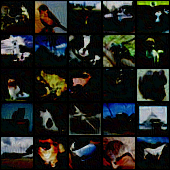}
    \caption{Attack using FGSM for CIFAR-10} \label{fig:fgsm_cifar}
\end{figure}

\section{Experiments for Defense}
\label{section5}
For defense, we denote the new classifier using the DLN trained against any attack $A$ or NAC modified classifier as $C'_{M}$ (for MNIST) or $C'_{C}$ (for CIFAR). Also, we test accuracies on test data adversarially perturbed by attacks against the new classifiers, for example, following our convention one such dataset would be denoted as $A(C'_{M},OTD)$. For defense we focus on defending against known attacks in literature.

\subsection{DLN Defense Against Single Attack}

Table~\ref{mnistsingleattack} shows the results when the DLN is trained to defense against specific attacks using MNIST dataset to yield a new classifier $C'_M$. Along expected lines, the accuracy on $OTD$ drops slightly for all the cases where DLN is trained against FGSM or CW.
However, the new classifier $C'_M$ is able to perform quite good on adversarial examples that were produced by any attack on the original classifier $C_M$, for example, $C'_M$ gets an accuracy of $95\%$ of the adversarial examples CW$({C_M},OTD)$ that was produced by CW attacking $C_M$. Lastly, when attacked again the new classifier $C'_M$ is still not resilient to attacks as shown by the low accuracies for FGSM$({C'_M},OTD)$ and CW$({C'_M},OTD)$. 

One number that stands out is the success of the new classifier $C'_M$ in correctly classifying the adversarial example generated by CW for the original classifier $C_M$. This supports our hypothesis that CW is very sensitive to the exact classifier boundary, and a newer classifier with a slightly different boundary is able to defeat prior adversarial examples. Of course, CW is is able to attack $C'_M$ against quite successfully, which we later show can  be defended by our NAC approach. For FGSM, we show in later sections that the performance of the classifier greatly improves when DLN is repeatedly trained against FGSM---revealing that the DLN approach is flexible enough to keep on improving its performance against high perturbation attacks.


\begin{table}[h!]
\centering
\begin{tabular}{p{1cm} c c c}
\toprule
DLN Trained & Test data type & Accuracy & Distortion\\
\midrule
FGSM& $OTD$ & 96.77 \% & $-$\\
FGSM & FGSM$({C_M},OTD)$ & 88.5 \% & 4.55 \\
FGSM & FGSM$({C'_M},OTD)$ & 13.75 \% & 6.98 \\
CW& $OTD$ & 98.6 \% & $-$\\
CW& CW$({C_M},OTD)$ & 95.42 \% & 5.77 \\
CW& CW$({C'_M},OTD)$ & 0.14 \% & 3.5 \\
\bottomrule
\end{tabular}
\caption{New DLN prepended classifier $C'_M$ for MNIST} \label{mnistsingleattack}
\end{table}


\subsection{Repeated DLN}
In this section, we run DLN repeatedly as described earlier in Section~\ref{repeatDLN}. We cut off the experiments when a single round took more than 24 hours to solve. We show the results for MNIST in Table~\ref{repeattraining} showing a clearly increasing trend in accuracy on adversarial examples produced by FGSM attacking the newer classifier, revealing increasing robustness. For CIFAR, the approach becomes too computationally expensive within two rounds. Thus, while the DLN approach is promising, as stated earlier it is computationally expensive, more so with larger and complex data sets. Further, as stated earlier running DLN against low perturbation attacks like CW does not show much improvement. However, we tackle that next using the NAC defense approach.

\begin{table}[h!]
\centering
\begin{tabular}{p{1cm}p{1.5cm}p{2.5cm}p{1.5cm}}
\toprule
Round & Acc. $OTD$ & Acc. FGSM$({C_i},OTD)$ & Distortion \\
\midrule
0 & 99.36 \% & 0.72  \% & 14.99 \\
1 & 97.70 \% & 13.70 \% & 13.63 \\
2 & 97.61 \% & 24.86 \% & 14.58 \\
3 & 97.95 \% & 43.39 \% & 14.73 \\
4 & 97.79 \% & 52.88 \% & 14.57 \\
5 & 97.77 \% & 56.57 \% & 14.52 \\
\bottomrule
\end{tabular}
\caption{Classifier trained repeatedly against FGSM for MNIST} \label{repeattraining}
\end{table}

\subsection{NAC defense}

Recall that the NAC defense works by adding noise to the logits layer of the neural network classifier to produce a new classifier $C'_M$ for MNIST and $C'_C$ for CIFAR. We use Gaussian noise with 0 mean and variance 1.
In this section, we show that the NAC defense is able to produce classifiers that are resilient to CW attack. Further, the new classifier's accuracy on the original test data-set is nearly unchanged. This can be seen in Table~\ref{nacmnist}. The second line in that table shows that the CW attack completely fails as the accuracy on the adversarial examples in $93\%$. However, it can also be observed that the new classifier is not resilient to attack by FGSM, as shown by  the third line in that table. This follows the intuition we provided in Figure~\ref{fig:dlnvsnac}. For CIFAR, Table~\ref{naccifar} shows that NAC is able to overcome CW to a large extent.

\begin{table}[h!]
\centering
\begin{tabular}{p{1cm}p{2.5cm}p{1.5cm}p{1.5cm}}
\toprule
Attack & Test data type & Accuracy & Distortion \\
\midrule
- & $OTD$ & 99.36  \% & - \\
CW & CW$({C'_M},OTD)$ & 93.60 \% & 1.49 \\
FGSM & FGSM$({C'_M},OTD)$ & 0.74 \% & 14.99 \\
\bottomrule
\end{tabular}
\caption{Accuracy of NAC Classifier $C'_M$ for MNIST} \label{nacmnist}
\end{table}

\begin{table}[h!]
\centering
\begin{tabular}{p{1cm}p{2.5cm}p{1.5cm}p{1.5cm}}
\toprule
Attack & Test data type & Accuracy & Distortion \\
\midrule
- & $OTD$ & 84.67  \% & - \\
CW & CW$({C'_M},OTD)$ & 77.70 \% & 0.17 \\
FGSM & FGSM$({C'_M},OTD)$ & 4.19 \% & 10.04 \\
\bottomrule
\end{tabular}
\caption{Accuracy of NAC Classifier $C'_M$ for CIFAR-10} \label{naccifar}
\end{table}

As stated earlier, a natural idea to attack NAC would be to query an image $n$ times and then average the logits before using it for the CW attack. This augmented attack does make CW more effective but not by much. Table~\ref{nacattackedmnist} shows that the accuracy on the adversarial example generated for $C'_M$ remains high. Moreover, more queries make it more difficult to conduct the CW attack in practice (as the adversary may be query limited), while also causing a small increase (1\% with 5000 sample) in the already high runtime of CW.

\begin{table}[h!]
\centering
\begin{tabular}{c c c}
\toprule
$n$ & Adversarial accuracy & Distortion \\
\midrule
500 &  95.14  \% & 1.51 \\
5000 &  82.07 \% & 1.51 \\
\bottomrule
\end{tabular}
\caption{Accuracy of NAC Classifier $C'_M$ against improved CW for MNIST} \label{nacattackedmnist}
\end{table}

\subsection{Defense Against Multiple Attacks}

Finally, we show that DLN and NAC can work together. We show this by presenting the accuracy on the adversarial example generated in each round of DLN repetition when the classifier $C_i$ after each round is augmented with NAC and attacked by FGSM and CW both. See Table ~\ref{repeattrainingNAC}. One observation is that NAC's performance decreases slightly over rounds stabilizing at 79\%, while the accuracy for original test set and FGSM perturbed test set stays almost exactly same as Table~\ref{repeattraining}.

\begin{table}[h!]
\centering
\begin{tabular}{p{1cm}p{1.5cm}p{2.5cm}p{1.5cm}}
\toprule
Round & Acc. $OTD$ & Acc. FGSM$({C_i},OTD)$ & Acc. CW$({C_i},OTD)$ \\
\midrule
0 & 99.36 \% & 0.72  \% & 94.2 \\
1 & 97.70 \% & 13.72 \% & 93.7 \\
2 & 97.73 \% & 24.28 \% & 84.7 \\
3 & 97.60 \% & 43.20 \% & 83.3 \\
4 & 97.64 \% & 53.17 \% & 79 \\
5 & 97.73 \% & 56.45 \% & 79.3 \\
\bottomrule
\end{tabular}
\caption{Classifier trained repeatedly against FGSM for MNIST and augmented with NAC in each round} \label{repeattrainingNAC}
\end{table}

\section{Related Work}
A thorough survey of security issues in machine learning is present in surveys~\cite{papernot2016towards} and some of the first results appeared in ~\cite{lowd2005adversarial,dalvi2004adversarial}. Here we discuss the most closely related work.

\textbf{Attacks}: Most previous attack work focuses on adversarial examples for computer vision tasks. Multiple techniques to create
such adversarial examples have been developed recently. Broadly, such attacks can be categorized as either using costs gradients~\cite{fgsm,deepfool,huang2015learning,biggio2013evasion} or the forward gradient of the neural network~\cite{limitation} and perturbing along most promising direction or 
directly solving an optimization problem (possibly using gradient ascent/descent) to find a perturbation~\cite{moosavi2017universal,carlini2017towards}.
In addition, adversarial examples have been
shown to transfer between different network architectures, and networks trained on disjoint subsets of
data~\cite{szegedy2013intriguing,papernot2016transferability}. Adversarial examples have also been shown to translate to the real world~\cite{kurakin2016adversarial}, that is, adversarial images can remain adversarial even after being printed and
recaptured with a cell phone camera.  Attacks on non-neural networks have also been explored in literature~\cite{biggio2013evasion}. Our approach is distinctly different from all these approaches as we pose the problem of generating adversarial samples as a generative learning problem, and demonstrate generation of adversarial examples given access to any given classifier. Our approach also applies to any classifier that output class probabilities and not just neural networks.

\textbf{Defense}: Also, referred to as robust classification in many papers, defense techniques can be roughly categorized into techniques that do 
(1) adversarial (re)training, which is adding back adversarial examples to the training data and retraining the classifier, often repeatedly~\cite{li2016general}, or modifying loss function to account for attacks~\cite{huang2015learning};
(2) gradient masking, which targets that gradient based attacks by trying to make the gradient less informative~\cite{papernot2016practical};
(3) input modification, which are techniques that modify (typically lower the dimension) the feature space of the input data to make crafting adversarial examples difficult~\cite{xu2017feature};
(4) game-theoretic formulation, which modifies the loss minimization problem as a constrained optimization with constraints provided by adversarial utility in performing perturbations~\cite{li2014feature}, and (5) filtering and de-noising, which aims to detect/filter or de-noise adversarial examples (cited below). 

Our defense approach differs from the first four kinds of defense as our DLN approach never modify the classifier or inputs but add a sanitizer (DLN) before the classifier. First, this increases the capacity of the resultant classifier $C'$, so that it can model more complex separators, which is not achieved when the classifier family stays the same. Further, our defense is agnostic to the type of attack and does not utilize properties of specific types of attacks. Interestingly, the DLN approach can be used with any classifier that output class probabilities and not just neural networks. Further, NAC is a very minor modification to the classifier that, distinct from other randomized approaches~\cite{vorobeychik2014optimal} that randomize over multiple classifiers, aims to mask the classifier boundary. Also, NAC can work with other defenses unlike techniques that modify inputs to try and defend against CW~\cite{xu2017feature}.

More closely related to our work are some defense techniques that have focused on detecting and filtering out adversarial samples~\cite{li2016adversarial,grosse2017statistical} or de-noising input~\cite{gu2014towards}; here the filter or de-noiser with the classifier could be considered as a larger neural network. However, unlike these work, our goal for DLN is targeted sanitization. Moreover, recent attack work~\cite{carlini2017adversarial} have produced attack techniques to defeat many known detection techniques. Our technique provides the flexibility to be resilient against more powerful attacks by training the DLN with such an attack for high perturbation attacks or using NAC for low perturbation attacks.

Lastly, two concurrent unpublished drafts (available online) have \emph{independently and simultaneously} proposed an attack~\cite{atn} similar to ALN and a defense~\cite{chen2016evaluation} apparently similar to DLN. The difference for the attack work is in using the class label vs classifier output in $\overline{opsim}$ term for the attack. For the defense work, we differ as we show how DLN technique extends to multiple attacks and can be repeatedly used in an attack-defense competition. Moreover, unlike these drafts, we provide another defense technique NAC that works against CW, define robustness and show that our defense techniques approximately aims to achieve our definition of robustness. Further, our formal reasoning reveals the underlying nature of attacks and defenses. 

\section{Conclusion and Future Work}
Our work provides a new learning perspective of the adversarial examples generation and defense problems with a formal intuition of how these approaches work, using which we were able to defend against multiple attacks including the potent CW. Further, unlike past work, our defense technique does not claim to a catchall or specific to any attack; in fact, it is flexible enough to possibly defend against any attack. Posing the attack and defense as learning problems allows for the possibility of using the rapidly developing research in machine learning itself to make the defense more effective in future, for example, by using a different specialized neural network architecture rather than an autoencoder.
 A number of variations of our theory and and tuning of the application framework provides rich content for future work.

\clearpage
\bibliographystyle{ACM-Reference-Format}
\bibliography{ref}


\begin{thebibliography}{00}


\ifx \showCODEN    \undefined \def \showCODEN     #1{\unskip}     \fi
\ifx \showDOI      \undefined \def \showDOI       #1{#1}\fi
\ifx \showISBNx    \undefined \def \showISBNx     #1{\unskip}     \fi
\ifx \showISBNxiii \undefined \def \showISBNxiii  #1{\unskip}     \fi
\ifx \showISSN     \undefined \def \showISSN      #1{\unskip}     \fi
\ifx \showLCCN     \undefined \def \showLCCN      #1{\unskip}     \fi
\ifx \shownote     \undefined \def \shownote      #1{#1}          \fi
\ifx \showarticletitle \undefined \def \showarticletitle #1{#1}   \fi
\ifx \showURL      \undefined \def \showURL       {\relax}        \fi
\providecommand\bibfield[2]{#2}
\providecommand\bibinfo[2]{#2}
\providecommand\natexlab[1]{#1}
\providecommand\showeprint[2][]{arXiv:#2}

\bibitem[\protect\citeauthoryear{Anthony and Bartlett}{Anthony and
  Bartlett}{2009}]%
        {Anthony2009NNL}
\bibfield{author}{\bibinfo{person}{Martin Anthony} {and}
  \bibinfo{person}{Peter~L. Bartlett}.} \bibinfo{year}{2009}\natexlab{}.
\newblock \bibinfo{booktitle}{{\em Neural Network Learning: Theoretical
  Foundations\/} (\bibinfo{edition}{1st} ed.)}.
\newblock \bibinfo{publisher}{Cambridge University Press},
  \bibinfo{address}{New York, NY, USA}.
\newblock
\showISBNx{052111862X, 9780521118620}


\bibitem[\protect\citeauthoryear{Baluja and Fischer}{Baluja and
  Fischer}{2017}]%
        {atn}
\bibfield{author}{\bibinfo{person}{Shumeet Baluja} {and} \bibinfo{person}{Ian
  Fischer}.} \bibinfo{year}{2017}\natexlab{}.
\newblock \showarticletitle{Adversarial Transformation Networks: Learning to
  Generate Adversarial Examples}.
\newblock \bibinfo{journal}{{\em CoRR\/}}  \bibinfo{volume}{abs/1703.09387}
  (\bibinfo{year}{2017}).
\newblock
\showURL{%
\url{http://arxiv.org/abs/1703.09387}}


\bibitem[\protect\citeauthoryear{Biggio, Corona, Maiorca, Nelson,
  {\v{S}}rndi{\'c}, Laskov, Giacinto, and Roli}{Biggio et~al\mbox{.}}{2013}]%
        {biggio2013evasion}
\bibfield{author}{\bibinfo{person}{Battista Biggio}, \bibinfo{person}{Igino
  Corona}, \bibinfo{person}{Davide Maiorca}, \bibinfo{person}{Blaine Nelson},
  \bibinfo{person}{Nedim {\v{S}}rndi{\'c}}, \bibinfo{person}{Pavel Laskov},
  \bibinfo{person}{Giorgio Giacinto}, {and} \bibinfo{person}{Fabio Roli}.}
  \bibinfo{year}{2013}\natexlab{}.
\newblock \showarticletitle{Evasion attacks against machine learning at test
  time}. In \bibinfo{booktitle}{{\em Joint European Conference on Machine
  Learning and Knowledge Discovery in Databases}}. Springer,
  \bibinfo{pages}{387--402}.
\newblock


\bibitem[\protect\citeauthoryear{Carlini and Wagner}{Carlini and
  Wagner}{2017a}]%
        {carlini2017adversarial}
\bibfield{author}{\bibinfo{person}{Nicholas Carlini} {and}
  \bibinfo{person}{David Wagner}.} \bibinfo{year}{2017}\natexlab{a}.
\newblock \showarticletitle{Adversarial Examples Are Not Easily Detected:
  Bypassing Ten Detection Methods}.
\newblock \bibinfo{journal}{{\em arXiv preprint arXiv:1705.07263\/}}
  (\bibinfo{year}{2017}).
\newblock


\bibitem[\protect\citeauthoryear{Carlini and Wagner}{Carlini and
  Wagner}{2017b}]%
        {carlini2017towards}
\bibfield{author}{\bibinfo{person}{Nicholas Carlini} {and}
  \bibinfo{person}{David Wagner}.} \bibinfo{year}{2017}\natexlab{b}.
\newblock \showarticletitle{Towards evaluating the robustness of neural
  networks}. In \bibinfo{booktitle}{{\em Security and Privacy (SP), 2017 IEEE
  Symposium on}}. IEEE, \bibinfo{pages}{39--57}.
\newblock


\bibitem[\protect\citeauthoryear{Chen, Li, and Vorobeychik}{Chen
  et~al\mbox{.}}{2016}]%
        {chen2016evaluation}
\bibfield{author}{\bibinfo{person}{Xinyun Chen}, \bibinfo{person}{Bo Li}, {and}
  \bibinfo{person}{Yevgeniy Vorobeychik}.} \bibinfo{year}{2016}\natexlab{}.
\newblock \showarticletitle{Evaluation of Defensive Methods for {DNNs} against
  Multiple Adversarial Evasion Models}.
\newblock  (\bibinfo{year}{2016}).
\newblock
\newblock
\shownote{\url{https://openreview.net/forum?id=ByToKu9ll&noteId=ByToKu9ll}.}


\bibitem[\protect\citeauthoryear{Cisse, Bojanowski, Grave, Dauphin, and
  Usunier}{Cisse et~al\mbox{.}}{2017}]%
        {cisse2017parseval}
\bibfield{author}{\bibinfo{person}{Moustapha Cisse}, \bibinfo{person}{Piotr
  Bojanowski}, \bibinfo{person}{Edouard Grave}, \bibinfo{person}{Yann Dauphin},
  {and} \bibinfo{person}{Nicolas Usunier}.} \bibinfo{year}{2017}\natexlab{}.
\newblock \showarticletitle{Parseval Networks: Improving Robustness to
  Adversarial Examples}.
\newblock \bibinfo{journal}{{\em arXiv preprint arXiv:1704.08847\/}}
  (\bibinfo{year}{2017}).
\newblock


\bibitem[\protect\citeauthoryear{Dalvi, Domingos, Sanghai, Verma,
  et~al\mbox{.}}{Dalvi et~al\mbox{.}}{2004}]%
        {dalvi2004adversarial}
\bibfield{author}{\bibinfo{person}{Nilesh Dalvi}, \bibinfo{person}{Pedro
  Domingos}, \bibinfo{person}{Sumit Sanghai}, \bibinfo{person}{Deepak Verma},
  {et~al\mbox{.}}} \bibinfo{year}{2004}\natexlab{}.
\newblock \showarticletitle{Adversarial classification}. In
  \bibinfo{booktitle}{{\em Proceedings of the tenth ACM SIGKDD international
  conference on Knowledge discovery and data mining}}. ACM,
  \bibinfo{pages}{99--108}.
\newblock


\bibitem[\protect\citeauthoryear{Fawzi, Fawzi, and Frossard}{Fawzi
  et~al\mbox{.}}{2015}]%
        {fawzi2015fundamental}
\bibfield{author}{\bibinfo{person}{Alhussein Fawzi}, \bibinfo{person}{Omar
  Fawzi}, {and} \bibinfo{person}{Pascal Frossard}.}
  \bibinfo{year}{2015}\natexlab{}.
\newblock \showarticletitle{Fundamental limits on adversarial robustness}. In
  \bibinfo{booktitle}{{\em Proc. ICML, Workshop on Deep Learning}}.
\newblock


\bibitem[\protect\citeauthoryear{Goodfellow, Pouget-Abadie, Mirza, Xu,
  Warde-Farley, Ozair, Courville, and Bengio}{Goodfellow
  et~al\mbox{.}}{2014a}]%
        {goodfellow2014generative}
\bibfield{author}{\bibinfo{person}{Ian Goodfellow}, \bibinfo{person}{Jean
  Pouget-Abadie}, \bibinfo{person}{Mehdi Mirza}, \bibinfo{person}{Bing Xu},
  \bibinfo{person}{David Warde-Farley}, \bibinfo{person}{Sherjil Ozair},
  \bibinfo{person}{Aaron Courville}, {and} \bibinfo{person}{Yoshua Bengio}.}
  \bibinfo{year}{2014}\natexlab{a}.
\newblock \showarticletitle{Generative adversarial nets}. In
  \bibinfo{booktitle}{{\em Advances in neural information processing systems}}.
  \bibinfo{pages}{2672--2680}.
\newblock


\bibitem[\protect\citeauthoryear{Goodfellow, Shlens, and Szegedy}{Goodfellow
  et~al\mbox{.}}{2014b}]%
        {fgsm}
\bibfield{author}{\bibinfo{person}{Ian~J. Goodfellow},
  \bibinfo{person}{Jonathon Shlens}, {and} \bibinfo{person}{Christian
  Szegedy}.} \bibinfo{year}{2014}\natexlab{b}.
\newblock \showarticletitle{Explaining and Harnessing Adversarial Examples}.
\newblock \bibinfo{journal}{{\em CoRR\/}}  \bibinfo{volume}{abs/1412.6572}
  (\bibinfo{year}{2014}).
\newblock
\showURL{%
\url{http://arxiv.org/abs/1412.6572}}


\bibitem[\protect\citeauthoryear{Grosse, Manoharan, Papernot, Backes, and
  McDaniel}{Grosse et~al\mbox{.}}{2017}]%
        {grosse2017statistical}
\bibfield{author}{\bibinfo{person}{Kathrin Grosse}, \bibinfo{person}{Praveen
  Manoharan}, \bibinfo{person}{Nicolas Papernot}, \bibinfo{person}{Michael
  Backes}, {and} \bibinfo{person}{Patrick McDaniel}.}
  \bibinfo{year}{2017}\natexlab{}.
\newblock \showarticletitle{On the (Statistical) Detection of Adversarial
  Examples}.
\newblock \bibinfo{journal}{{\em arXiv preprint arXiv:1702.06280\/}}
  (\bibinfo{year}{2017}).
\newblock


\bibitem[\protect\citeauthoryear{Gu and Rigazio}{Gu and Rigazio}{2014}]%
        {gu2014towards}
\bibfield{author}{\bibinfo{person}{Shixiang Gu} {and} \bibinfo{person}{Luca
  Rigazio}.} \bibinfo{year}{2014}\natexlab{}.
\newblock \showarticletitle{Towards deep neural network architectures robust to
  adversarial examples}.
\newblock \bibinfo{journal}{{\em arXiv preprint arXiv:1412.5068\/}}
  (\bibinfo{year}{2014}).
\newblock


\bibitem[\protect\citeauthoryear{Huang, Xu, Schuurmans, and
  Szepesv{\'a}ri}{Huang et~al\mbox{.}}{2015}]%
        {huang2015learning}
\bibfield{author}{\bibinfo{person}{Ruitong Huang}, \bibinfo{person}{Bing Xu},
  \bibinfo{person}{Dale Schuurmans}, {and} \bibinfo{person}{Csaba
  Szepesv{\'a}ri}.} \bibinfo{year}{2015}\natexlab{}.
\newblock \showarticletitle{Learning with a strong adversary}.
\newblock \bibinfo{journal}{{\em arXiv preprint arXiv:1511.03034\/}}
  (\bibinfo{year}{2015}).
\newblock


\bibitem[\protect\citeauthoryear{Kurakin, Goodfellow, and Bengio}{Kurakin
  et~al\mbox{.}}{2016}]%
        {kurakin2016adversarial}
\bibfield{author}{\bibinfo{person}{Alexey Kurakin}, \bibinfo{person}{Ian
  Goodfellow}, {and} \bibinfo{person}{Samy Bengio}.}
  \bibinfo{year}{2016}\natexlab{}.
\newblock \showarticletitle{Adversarial examples in the physical world}.
\newblock \bibinfo{journal}{{\em arXiv preprint arXiv:1607.02533\/}}
  (\bibinfo{year}{2016}).
\newblock


\bibitem[\protect\citeauthoryear{Li and Vorobeychik}{Li and
  Vorobeychik}{2014}]%
        {li2014feature}
\bibfield{author}{\bibinfo{person}{Bo Li} {and} \bibinfo{person}{Yevgeniy
  Vorobeychik}.} \bibinfo{year}{2014}\natexlab{}.
\newblock \showarticletitle{Feature cross-substitution in adversarial
  classification}. In \bibinfo{booktitle}{{\em Advances in neural information
  processing systems}}. \bibinfo{pages}{2087--2095}.
\newblock


\bibitem[\protect\citeauthoryear{Li, Vorobeychik, and Chen}{Li
  et~al\mbox{.}}{2016}]%
        {li2016general}
\bibfield{author}{\bibinfo{person}{Bo Li}, \bibinfo{person}{Yevgeniy
  Vorobeychik}, {and} \bibinfo{person}{Xinyun Chen}.}
  \bibinfo{year}{2016}\natexlab{}.
\newblock \showarticletitle{A General Retraining Framework for Scalable
  Adversarial Classification}.
\newblock \bibinfo{journal}{{\em arXiv preprint arXiv:1604.02606\/}}
  (\bibinfo{year}{2016}).
\newblock


\bibitem[\protect\citeauthoryear{Li and Li}{Li and Li}{2016}]%
        {li2016adversarial}
\bibfield{author}{\bibinfo{person}{Xin Li} {and} \bibinfo{person}{Fuxin Li}.}
  \bibinfo{year}{2016}\natexlab{}.
\newblock \showarticletitle{Adversarial Examples Detection in Deep Networks
  with Convolutional Filter Statistics}.
\newblock \bibinfo{journal}{{\em arXiv preprint arXiv:1612.07767\/}}
  (\bibinfo{year}{2016}).
\newblock


\bibitem[\protect\citeauthoryear{Lowd and Meek}{Lowd and Meek}{2005}]%
        {lowd2005adversarial}
\bibfield{author}{\bibinfo{person}{Daniel Lowd} {and}
  \bibinfo{person}{Christopher Meek}.} \bibinfo{year}{2005}\natexlab{}.
\newblock \showarticletitle{Adversarial learning}. In \bibinfo{booktitle}{{\em
  Proceedings of the eleventh ACM SIGKDD international conference on Knowledge
  discovery in data mining}}. ACM, \bibinfo{pages}{641--647}.
\newblock


\bibitem[\protect\citeauthoryear{Moosavi{-}Dezfooli, Fawzi, and
  Frossard}{Moosavi{-}Dezfooli et~al\mbox{.}}{2015}]%
        {deepfool}
\bibfield{author}{\bibinfo{person}{Seyed{-}Mohsen Moosavi{-}Dezfooli},
  \bibinfo{person}{Alhussein Fawzi}, {and} \bibinfo{person}{Pascal Frossard}.}
  \bibinfo{year}{2015}\natexlab{}.
\newblock \showarticletitle{DeepFool: a simple and accurate method to fool deep
  neural networks}.
\newblock \bibinfo{journal}{{\em CoRR\/}}  \bibinfo{volume}{abs/1511.04599}
  (\bibinfo{year}{2015}).
\newblock
\showURL{%
\url{http://arxiv.org/abs/1511.04599}}


\bibitem[\protect\citeauthoryear{Moosavi~Dezfooli, Fawzi, Fawzi, and
  Frossard}{Moosavi~Dezfooli et~al\mbox{.}}{2017}]%
        {moosavi2017universal}
\bibfield{author}{\bibinfo{person}{Seyed~Mohsen Moosavi~Dezfooli},
  \bibinfo{person}{Alhussein Fawzi}, \bibinfo{person}{Omar Fawzi}, {and}
  \bibinfo{person}{Pascal Frossard}.} \bibinfo{year}{2017}\natexlab{}.
\newblock \showarticletitle{Universal adversarial perturbations}. In
  \bibinfo{booktitle}{{\em Proceedings of 2017 IEEE Conference on Computer
  Vision and Pattern Recognition (CVPR)}}.
\newblock


\bibitem[\protect\citeauthoryear{Papernot, McDaniel, and Goodfellow}{Papernot
  et~al\mbox{.}}{2016a}]%
        {papernot2016transferability}
\bibfield{author}{\bibinfo{person}{Nicolas Papernot}, \bibinfo{person}{Patrick
  McDaniel}, {and} \bibinfo{person}{Ian Goodfellow}.}
  \bibinfo{year}{2016}\natexlab{a}.
\newblock \showarticletitle{Transferability in machine learning: from phenomena
  to black-box attacks using adversarial samples}.
\newblock \bibinfo{journal}{{\em arXiv preprint arXiv:1605.07277\/}}
  (\bibinfo{year}{2016}).
\newblock


\bibitem[\protect\citeauthoryear{Papernot, McDaniel, Goodfellow, Jha, Celik,
  and Swami}{Papernot et~al\mbox{.}}{2016b}]%
        {papernot2016practical}
\bibfield{author}{\bibinfo{person}{Nicolas Papernot}, \bibinfo{person}{Patrick
  McDaniel}, \bibinfo{person}{Ian Goodfellow}, \bibinfo{person}{Somesh Jha},
  \bibinfo{person}{Z~Berkay Celik}, {and} \bibinfo{person}{Ananthram Swami}.}
  \bibinfo{year}{2016}\natexlab{b}.
\newblock \showarticletitle{Practical black-box attacks against deep learning
  systems using adversarial examples}.
\newblock \bibinfo{journal}{{\em arXiv preprint arXiv:1602.02697\/}}
  (\bibinfo{year}{2016}).
\newblock


\bibitem[\protect\citeauthoryear{Papernot, McDaniel, Jha, Fredrikson, Celik,
  and Swami}{Papernot et~al\mbox{.}}{2016c}]%
        {limitation}
\bibfield{author}{\bibinfo{person}{Nicolas Papernot}, \bibinfo{person}{Patrick
  McDaniel}, \bibinfo{person}{Somesh Jha}, \bibinfo{person}{Matt Fredrikson},
  \bibinfo{person}{Z~Berkay Celik}, {and} \bibinfo{person}{Ananthram Swami}.}
  \bibinfo{year}{2016}\natexlab{c}.
\newblock \showarticletitle{The limitations of deep learning in adversarial
  settings}. In \bibinfo{booktitle}{{\em Security and Privacy (EuroS\&P), 2016
  IEEE European Symposium on}}. IEEE, \bibinfo{pages}{372--387}.
\newblock


\bibitem[\protect\citeauthoryear{Papernot, McDaniel, Sinha, and
  Wellman}{Papernot et~al\mbox{.}}{2016d}]%
        {papernot2016towards}
\bibfield{author}{\bibinfo{person}{Nicolas Papernot}, \bibinfo{person}{Patrick
  McDaniel}, \bibinfo{person}{Arunesh Sinha}, {and} \bibinfo{person}{Michael
  Wellman}.} \bibinfo{year}{2016}\natexlab{d}.
\newblock \showarticletitle{Towards the Science of Security and Privacy in
  Machine Learning}.
\newblock \bibinfo{journal}{{\em arXiv preprint arXiv:1611.03814\/}}
  (\bibinfo{year}{2016}).
\newblock


\bibitem[\protect\citeauthoryear{Szegedy, Zaremba, Sutskever, Bruna, Erhan,
  Goodfellow, and Fergus}{Szegedy et~al\mbox{.}}{2013}]%
        {szegedy2013intriguing}
\bibfield{author}{\bibinfo{person}{Christian Szegedy},
  \bibinfo{person}{Wojciech Zaremba}, \bibinfo{person}{Ilya Sutskever},
  \bibinfo{person}{Joan Bruna}, \bibinfo{person}{Dumitru Erhan},
  \bibinfo{person}{Ian Goodfellow}, {and} \bibinfo{person}{Rob Fergus}.}
  \bibinfo{year}{2013}\natexlab{}.
\newblock \showarticletitle{Intriguing properties of neural networks}.
\newblock \bibinfo{journal}{{\em arXiv preprint arXiv:1312.6199\/}}
  (\bibinfo{year}{2013}).
\newblock


\bibitem[\protect\citeauthoryear{Wang, Gao, and Qi}{Wang et~al\mbox{.}}{2016}]%
        {wang2016theoretical}
\bibfield{author}{\bibinfo{person}{Beilun Wang}, \bibinfo{person}{Ji Gao},
  {and} \bibinfo{person}{Yanjun Qi}.} \bibinfo{year}{2016}\natexlab{}.
\newblock \showarticletitle{A Theoretical Framework for Robustness of (Deep)
  Classifiers Under Adversarial Noise}.
\newblock \bibinfo{journal}{{\em arXiv preprint arXiv:1612.00334\/}}
  (\bibinfo{year}{2016}).
\newblock


\bibitem[\protect\citeauthoryear{Xu, Evans, and Qi}{Xu et~al\mbox{.}}{2017}]%
        {xu2017feature}
\bibfield{author}{\bibinfo{person}{Weilin Xu}, \bibinfo{person}{David Evans},
  {and} \bibinfo{person}{Yanjun Qi}.} \bibinfo{year}{2017}\natexlab{}.
\newblock \showarticletitle{Feature Squeezing: Detecting Adversarial Examples
  in Deep Neural Networks}.
\newblock \bibinfo{journal}{{\em arXiv preprint arXiv:1704.01155\/}}
  (\bibinfo{year}{2017}).
\newblock


\end{thebibliography}

\clearpage


\end{document}